\documentclass[12pt]{article}
\pdfoutput=1

\title{Expected capture time and throttling number for cop versus gambler}
\author{Jesse Geneson\thanks{Department of Mathematics, Iowa State University,
Ames, IA 50011, USA (geneson@gmail.com).}
\and Carl Joshua Quines\thanks{cj@cjquines.com} 
\and Espen Slettnes\thanks{Abel Academy (Homeschool); UC Berkeley Extension, Berkeley, CA 94720, USA (eslettnes@berkeley.edu)}
\and Shen-Fu Tsai\thanks{parity@gmail.com}}
\date{}

\usepackage[inline]{asymptote}
\usepackage{color}
\usepackage{amsmath}
\usepackage{amsfonts}
\usepackage{amssymb}
\usepackage{amsthm}
\usepackage{mathtools}
\usepackage{mathrsfs}
\usepackage{graphicx}
\usepackage[hidelinks]{hyperref}
\usepackage{float}
\usepackage{cancel}
\usepackage[none]{hyphenat}
\usepackage{lastpage}
\usepackage{tcolorbox}
\usepackage[justification=centering]{caption}
\usepackage[backend=biber,style=numeric]{biblatex}
\bibliography{gthre}

\usepackage{geometry}
\usepackage{enumitem}
\setlist{itemsep=0pt,topsep=0pt}

\makeatletter
\let\thetitle        \@title
\let\theauthor       \@author
\let\thedate         \@date
\makeatother

\usepackage{fancyhdr}
\newtheorem{lem}{Lemma}

\theoremstyle{definition}
\newtheorem{thm}[lem]{Theorem}
\newtheorem{remark}[lem]{Remark}
\newtheorem{definition}[lem]{Definition}

\newtheorem{cor}[lem]{Corollary}

\newcommand{\pr}{\mathbf{Pr}}
\newcommand{\e}{\mathbf{E}}
\newcommand{\ut}[1] {{\upsilon\left(#1\right)}}
\newcommand{\kt}[1] {{\kappa\left(#1\right)}}
\newcommand{\ot}[2][] {{o_{#1}\left(#2\right)}}
\newcommand{\wmw}[1] {{\mathrm{WMW}\left(#1\right)}}
\newcommand{\kw}[1] {{\mathrm{KW}\left(#1\right)}}
\newcommand{\ceil}[1] {{\left\lceil{#1}\right\rceil}}
\newcommand{\floor}[1] {{\left\lfloor{#1}\right\rfloor}}
\newcommand{\abs}[1] {{\left|{#1}\right|}}

\renewcommand{\pmod}[1]{\,\left(\mathrm{mod\,}{#1}\right)}

\begin{document}
\maketitle

\begin{abstract}
We bound expected capture time and throttling number for the cop versus gambler game on a connected graph with $n$ vertices, a variant of the cop versus robber game that is played in darkness, where the adversary hops between vertices using a fixed probability distribution. The paper that originally defined the cop versus gambler game focused on two versions, a known gambler whose distribution the cop knows, and an unknown gambler whose distribution is secret. We define a new version of the gambler where the cop makes a fixed number of observations before the lights go out and the game begins. We show that the strategy that gives the best possible expected capture time of $n$ for the known gambler can also be used to achieve nearly the same expected capture time against the observed gambler when the cop makes a sufficiently large number of observations. We also show that even with only a single observation, the cop is able to achieve an expected capture time of approximately $1.5n$, which is much lower than the expected capture time of the best known strategy against the unknown gambler (approximately $1.95n$).
\end{abstract}

\section*{Introduction}
Graph pursuit has been investigated for over a century \cite{dudeney}, with most progress in the last few decades focusing on the game of cop versus robber \cite{berarducci, bonato, hahn1, hahn2}, starting with \cite{robber0, robber0q}. The main problems that have been investigated are whether the cop wins a graph if both cop and robber play optimally, how many cops are necessary to win a graph, and how long it takes to capture the robber.

Both the cop and robber take turns moving from vertex to adjacent vertex, or staying in place, so they are allowed to make the exact same set of moves, but the robber tries to evade and the cop tries to pursue. Other adversaries besides the robber have also been investigated, such as the rabbit \cite{adler, babichenko} which is allowed to jump to non-adjacent vertices. 

Komarov and Winkler introduced the gambler in \cite{gambler0}. The cop and gambler play on a connected $n$-vertex graph. First, the cop chooses an intitial position; then, the gambler chooses a probability distribution $p_1, \dots, p_n$. Every turn, the cop moves to an adjacent vertex or stays put, while simultaneously the gambler moves to vertex $v_i$ with probability $p_i$. The gambler is called \emph{known} if the cop knows his distribution before the cop takes his first turn. Otherwise the gambler is called \emph{unknown}.

The gambler was motivated by a problem in software design \cite{gambler0}. Suppose that an anti-incursion program must navigate a linked list of ports in order to minimize the interception time for an enemy packet. The known gambler models the case when the enemy's port distribution is known, while the unknown gambler models the case when the port distribution is unknown.

Komarov and Winkler proved that if both the cop and known gambler use best play, then the known gambler has expected capture time $n$ for any connected $n$-vertex graph \cite{gambler0}. They also found a pursuit algorithm for all connected $n$-vertex graphs that has expected capture time at most approximately $1.97n$ for the unknown gambler. This upper bound for the unknown gambler was improved to approximately $1.95n$ in \cite{jg1}. For the case of $n$-vertex cycle graphs, Komarov found sharp upper and lower bounds of approximately $1.082n$ for the unknown gambler \cite{komarov}. Komarov and Winkler conjectured that the upper bound for the unknown gambler can be improved to $3n/2$, with the star being the worst case \cite{gambler0}.

The game can be adapted to $k$ or more cops, where two cops can occupy the same vertex, one cop is needed to catch the gambler, and all cops can communicate with each other. Similarly to the original version, the gambler is \emph{known} if all cops know the distribution chosen by the gambler, and \emph{unknown} otherwise.

We introduce a version of the gambler in between the known and unknown gamblers. In this version, the cop does not know the distribution, but the gambler is visible for $t$ turns before the game starts (during which the cop is frozen at their initial position and cannot catch the gambler). The cop is also restricted to their initial position in the first turn of the game. We call this the \emph{$t$-observed gambler}. This version of the gambler is motivated by the problem of designing an anti-incursion program for an enemy whose port distribution is unknown, except for the knowledge of a given number of ports that the enemy has accessed so far.

Note that our definition of the $t$-observed gambler is not the same as if we allow the cop to see the gambler for the first $t$ moves of the game, and then we turn the lights out. We discuss this alternative variation of the observed gambler at the end of the paper. In particular, we show in the conclusion that the expected capture time for an $n$-vertex star is exactly $n$ if the observation happens in-game. 

In addition to expected capture time, we also investigate the throttling number for the cop versus gambler game. Throttling minimizes the sum of the number of cops and the expected capture time. The throttling number has previously been investigated for the cop versus robber game \cite{cvr} and zero forcing \cite{butler}.

In Section \ref{s:wmw}, we present an algorithm called Watch-Move-Wait ($\wmw{1}$). Against the $1$-observed gambler, on a connected $n$-vertex graph of radius $r$, our algorithm gives expected capture time at most $n+r$. As a corollary, connected $n$-vertex graphs have expected capture time at most $3n/2$ and $n$-vertex stars have expected capture time at most $n+1$ for the $1$-observed gambler.

In Section \ref{s:gwmw}, we generalize $\wmw{1}$ to a family of strategies $\wmw{t}$, and introduce another family of strategies called Komarov--Winkler ($\kw{t}$), based on an algorithm in \cite{gambler0}. Against the $t$-observed gambler, we prove that $t$ must be at least $\Omega(\sqrt{n})$ for either generalized $\wmw{t}$ or $\kw{t}$ to perform substantially better than $\wmw{1}$. Also, we find that for $t = \omega(n^2)$, $\kw{t}$ works almost as well against the $t$-observed gambler as the known gambler.

In Section \ref{s:dist}, we modify $\wmw{1}$ for when there are $k$ cops pursuing the $1$-observed gambler. For this case, we show that the $1$-observed gambler has expected capture time at most $3n/k-1$ in general, and $n/k+2$ for $n$-vertex stars. The general bound improves on the bound of $3.94n/k+O(1)$ from \cite{jg2}.

In Section \ref{s:throttling}, we find $\Theta\left(\sqrt n\right)$ bounds on the throttling numbers of the known, unknown, and observed gamblers for all connected graphs $G$ on $n$ vertices, as well as sharper bounds for families of graphs. In \cite{cvr}, it was conjectured that the maximum possible throttling number for the cop versus robber game on a connected $n$-vertex graph is also $\Theta\left(\sqrt n\right)$. However, this conjecture was later refuted in \cite{cvr2}, which exhibited connected $n$-vertex graphs with throttling number of $\Omega(n^{2/3})$ for the cop versus robber game. Thus our results show that the gambler has a much lower maximum throttling number than the robber on connected $n$-vertex graphs.

\section{Watch-Move-Wait} \label{s:wmw}

We use $\wmw{1}$ to denote the strategy where the cop watches the gambler for a single turn before the game begins, then moves from their initial position to wherever the gambler appeared, and waits there until she catches the gambler.

\begin{thm}\label{1-obs}
    Against the $1$-observed gambler on a connected $n$-vertex graph of radius $r$, the $\wmw{1}$ strategy has expected capture time at most $n+r$.
\end{thm}

\begin{proof}
Let $G$ be a connected $n$-vertex graph of radius $r$ and suppose that the cop has initial position at a center vertex $c$ of $G$, so every other vertex in $G$ has distance at most $r$ to $c$. The probability that the gambler appears at $v_i$ on the observed turn is $p_i$, so the cop chooses $v_i$ as their final vertex with probability $p_i$. 

If $p_i > 0$ and the cop chooses final vertex $v_i$, then the expected capture time is at most $r+\frac{1}{p_i}$. Here, $r$ is the maximum possible distance to move to $v_i$, and $\frac{1}{p_i}$ is the expected capture time after the cop waits at $v_i$. 

Let $S = \left\{i: 1\leq i \leq n \wedge p_i > 0  \right\}$, the set of indices of all vertices with positive probability. By summing over $S$, the expected capture time is at most 
$$\sum\limits_{i \in S} p_i \left(r+\frac{1}{p_i}\right) = r+\abs{S} \leq r+n.$$
\end{proof}

Two corollaries follow from this. From the fact that every connected $n$-vertex graph has radius at most $n/2$, we find that:

\begin{cor}\label{upgen}
$\wmw{1}$ has expected capture time at most $3n/2$ against the $1$-observed gambler on any connected $n$-vertex graph. 
\end{cor}

As stars have radius $1$, we also find that:

\begin{cor}
If both cop and gambler use best play, the expected capture time for the $1$-observed gambler on the $n$-vertex star is at most $n+1$. 
\end{cor}

The last upper bound is nearly tight, since $n$-vertex stars have expected capture time of $n$ for the known gambler, and the expected capture time of the known gambler is at most the expected capture time of the $1$-observed gambler.

\section{Generalized Watch-Move-Wait} \label{s:gwmw}
We use $\wmw{t}$ to denote the strategy where the cop observes $t$ turns and then goes to the modal vertex. If there is a tie, the cop goes to the closest mode. If there is still a tie, the cop picks among the tied vertices uniformly at random.

We also define a family of strategies using the algorithm defined by Komarov and Winkler in \cite{gambler0} for the known gambler. We call this the Komarov--Winkler (KW) strategy.

Define $\kw{t}$ as the strategy where the cop observes $t$ turns, calculates sample probabilities $p'_{i}$ for each vertex, and uses the Komarov--Winkler algorithm with the sample probabilities.

The next result shows that $t$ must be at least $\Omega(\sqrt{n})$ for either generalized WMW or KW to perform substantially better than with a single observation (i.e. $t = 1$).

\begin{thm}
For $t = o(\sqrt{n})$, there exists a gambler distribution on the cycle $C_{n}$ such that the expected capture time for the $t$-observed gambler is $\left(1.5-o(1)\right)n$ when the cop uses $\wmw{t}$ or $\kw{t}$. 
\end{thm}

\begin{proof}
The same proof works for both $\wmw{t}$ and $\kw{t}$. 

Choose a sequence $\epsilon_n = O(\frac{1}{\sqrt{n}})$. If the cop starts at vertex $v$ on $C_n$, let $R$ be the set of the $r = \floor{t\sqrt{n}}$ or $\floor{t\sqrt{n}} + 1$ vertices farthest from $v$, where the choice of $r$ is determined by the parity of $n$. For $n$ sufficiently large so that $1-\epsilon_n$ and $n-r$ are positive, the gambler chooses the distribution with probabilities $\frac{1-\epsilon_n}{r}$ for each vertex in $R$ and $\frac{\epsilon_n}{n-r}$ for each remaining vertex.

Let $v_1, \dots, v_n$ be the vertices on the cycle and let $u_1, \dots, u_t$ be the observed vertices in order. Let $T$ be the random variable for the capture time. Let $X$ be the event that all observations are distinct, $Y$ the event that all observations are vertices in $R$, and $Z$ the event that the first observation $u_1$ is \emph{not} in $R$. 

Note that $$\pr[X \wedge Y] \geq \prod_{i = 1}^{t} \left(1 - \epsilon_n - \frac{i}{r}\right) \geq 1 - t\epsilon_n - \frac{t^2}{r} = 1 - o(1).$$ 

Moreover, as $$\pr[X \mid u_1 = v_i] \geq \prod_{j = 1}^{t} \left(1-\frac{j}{r}\right) \geq 1-\frac{t^2}{r} = 1-o(1)$$for all $i$, we get $\pr[X \wedge Z] \geq (1-o(1))\epsilon_n$ by summing over the vertices not in $R$. As $Y$ and $Z$ are disjoint events, the expected capture time is at least
\begin{align*}
  \e[T] &\geq \pr[X \wedge Y]\e[T \mid X \wedge Y] + \pr[X \wedge Z]\, \e[T \mid X \wedge Z] \\
  & \geq \left(1-o(1)\right)\left(1-\frac{\epsilon_n}{n-r}\right)^{\frac{n}{2}}\left(\frac{n}{2}-o(n)\right) + \left(1 - o(1)\right)\epsilon_n \, \frac{1}{\frac{\epsilon_n}{n-r}} = \left(1.5-o(1)\right)n.
\end{align*}
\end{proof}

The next three results show that when $t = \omega(n^2)$, $\kw{t}$ works almost as well against the $t$-observed gambler as the known gambler. 

\begin{lem}
    \label{lem:samplerror}
    Suppose that $p_1, \dots, p_n$ is a discrete probability distribution with $n$ outcomes that outputs $i$ with probability $p_i$. If $p'_i$ is the fraction of $W$ samples equal to $i,$ then for any $P\in(0,1)$
    $$\abs{p'_i - p_i} \le \sqrt {\frac{1/W}{1-P}}\,$$
    for every $ i\in \lbrace 1, \ldots, n \rbrace$ with probability at least $P.$
\end{lem}
\begin{proof}
    The variance in $p'_i$ is $\frac{p_i(1-p_i)}W.$ By Chebyshev's inequality, with probability at least $1-p_i(1-P),$ 
    \[
    \abs{p'_i-p_i}
    \le \sqrt{\frac{\frac{p_i(1-p_i)}W}{1 - \left(1-p_i(1-P)\right)}}
    = \sqrt{\frac{\frac{(1-p_i)}W}{1-P}}
    \le \sqrt{\frac{1/W}{1-P}}.
    \]
    
    By the intersection bound, we then have $\abs{p'_i - p_i} \le \sqrt {\frac{1/W}{1-P}}$ for all $i$, with probability at least $1 - \sum \limits_{i=1}^n p_i \left(1 - P\right) = 1 - (1-P) = P.$ 
\end{proof}

We note that the bound in the next theorem can be obtained by a similar proof to that of the upper bound of $n$ for the known gambler in \cite{gambler0}, by using the sample probabilities in place of the known probabilities and adjusting a few of the inequalities. However, we found that the proof is simpler if we use slightly different values in place of the known probabilities. Even though these values do not add to $1$, the proof still yields the bound below, following exactly the same steps as in \cite{gambler0}.

\begin{thm}\label{kwthmp}
    For any connected $n$-vertex graph and initial cop position, there exists a strategy against the $w n^2$-observed gambler with an expected capture time of at most
    $$\left(1- \sqrt{\frac{1/w}{1-P}}\right)^{-1} n,$$
    with probability at least $P$, for any $P < \frac{w-1}{w}.$
\end{thm}

\begin{proof}
    After $wn^2$ observations, let $p_i$ be the probability the gambler assigns to vertex $i,$ and let $p'_i$ be the fraction of observations for which the gambler appeared at $i$. By Lemma~\ref{lem:samplerror}, if $\epsilon = \sqrt{\frac{1/(wn^2)}{1-P}},$ we have $\abs{p'_i - p_i} \le \epsilon$ for all $i$, with probability at least $P.$
    
    Now the cop uses the algorithm from the proof of Lemma 2.1 in \cite{gambler0}, replacing the probabilities with those of the form $\max(0, p'_i - \epsilon)$. Following the same proof in \cite{gambler0}, with probability at least $P$ the expected capture time is at most $\dfrac{n}{1 - n\epsilon} = \left(1- \sqrt{\frac{1/w}{1-P}}\right)^{-1} n$.
\end{proof}

\begin{thm}
Suppose that $\beta(n)$ is a positive function such that $\lim_{n \rightarrow \infty} \beta(n) = \infty$ and let $t = t_n \geq n^2 \beta(n)^2$. For sufficiently large $n,$ there is a strategy to catch the $t$-observed gamber with expected capture time at most $\left(\left(1-\frac{1}{\beta(n)^{1/4}}\right)^{-1}+\frac{2}{\beta(n)^{1/2}}+\frac{4}{\beta(n)^{3/2}}\right) n$ on any connected $n$-vertex graph.
\end{thm}

\begin{proof}
Define $x = n \beta(n)$. We modify the strategy from Theorem~\ref{kwthmp} as follows. We use the same strategy for the first $x$ turns. If, after turn $x$, the gambler has not been caught, we switch to Komarov and Winkler's strategy for catching the unknown gambler.

By applying Lemma~\ref{lem:samplerror}, we find that with probability at least $P = 1-\frac{1}{\beta(n)^{3/2}}$, the cop either catches the gambler or reaches a vertex $v_i$ with $p_i \geq \frac{1}{n}(1-\frac{1}{\beta(n)^{1/4}}) > \frac{1}{2n}$ within $n$ turns, for sufficiently large $n$. Thus with probability at least $P(1-(1-\frac{1}{2n})^{x-n}) \geq P(1-e^{-\beta(n)/2+1/2})$, the gambler will be caught using the strategy in Theorem~\ref{kwthmp} in at most $x$ turns.

Suppose the gambler is not caught in at most $x$ turns. By \cite{gambler0}, using the modified strategy, the expected capture time is less than $x + 2n$. Thus, for sufficiently large $n,$ the expected capture time for the modified strategy is at most $$\left(1-\frac{1}{\beta(n)^{1/4}}\right)^{-1}n+\frac{2}{\beta(n)^{3/2}}(x+2n) \leq \left(\left(1-\frac{1}{\beta(n)^{1/4}}\right)^{-1}+\frac{2}{\beta(n)^{1/2}}+\frac{4}{\beta(n)^{3/2}}\right) n.$$
\end{proof}

\section{Observed gambler versus $k$ cops} \label{s:dist}

In \cite{jg2}, Geneson defined an algorithm to partition a connected graph into sectors, all of whose sizes are within a factor of $2$ of each other. That paper used the algorithm in a strategy to pursue the unknown gambler with expected capture time at most $3.94n/k+O(1)$. We use the same sectoring algorithm for the results in this section, so it is presented below for completeness.

We say that a \textsl{branch} of a rooted tree is an induced rooted subtree formed by a vertex and one of its children's subtrees. A \textsl{limb} is a union of a vertex and some (possibly all or none) of its branches.

\begin{lem}\label{lem:sector}
If $T = (V,E)$ is a rooted tree with $\abs{V} = n$ and $x$ is an integer such that $2 \leq x < n$, then there exists a subset $S \subset V$ and a vertex $v \in S$ such that $x < \abs{S} \leq 2x-1$, $T|_{S}$ is a limb of $T$, and $T|_{(V-S) \cup \left\{v\right\} }$ is connected.
\end{lem}

\begin{proof}
We use strong induction on $n.$ For the base case, if $x < n \leq 2x-1$, $\forall v \in V,\, S = V$ works. For $n > 2x-1$, the inductive hypothesis is that the lemma is true for all $m$ less than $n$.

Suppose that $u$ is the root of $T$, and let $B$ be a branch of $T$ rooted at $u$ with a maximum number of vertices among all branches rooted at $u$. If $\abs{B} > 2x-1$, the inductive hypothesis on $B - \left\{u \right\}$ gives values for $S$ and $v$. If $x < \abs{B} \leq 2x-1$, then set $S = B$ and $v = u$.

If $\abs{B} \leq x$, set $C = B$ and add all vertices to $C$ from each next largest branch rooted at $u$ until $C$ has more than $x$ vertices. Since $B$ has the maximum number of vertices among all branches rooted at $u,$ at most $\abs{B}-1 \le x-1$ vertices are added at a time, so this process results in $x < \abs{C} \leq 2x-1$. Set $S = C$ and $v = u$.
\end{proof}

\begin{minipage}{0.66\textwidth}
	\begin{remark}
		In the preceding lemma, $2x-1$ is tight. In the tree depicted in Figure~\ref{fig:trees}, there is no such $S, v \in S$ where $x < \abs{S} < 2x-1.$
	\end{remark}
	We now show how to cover a graph's vertex set with connected subgraphs of small size.
\end{minipage}
\begin{minipage}{0.33\textwidth}
	\begin{figure}[H]
		\centering
		\begin{asy}
			unitsize(0.667inch);
			
			pair[] childPositions (int numChildren, pair offset=(0,0), real length=0.5, real separation=pi/15)
			{
			  pair[] positions;
			  for (int i=0; i<numChildren; ++i)
			    {
			      positions.push(offset + length*expi(3*pi/2 + separation*(i-(numChildren-1)/2.0)));
			    }
			  return positions;
			}
			
			pair[] vertices;
			pair O = (0,0);
			
			vertices.push(O);
			dot(O);
			
			pair[] firstgen = { (-1,-0.5), (0,-0.5), (1,-0.5) };
			vertices.append(firstgen);
			
			for (int i=0; i<firstgen.length; ++i)
			  {
			    draw(O -- firstgen[i]);
			    dot(firstgen[i]);
			    pair[] secondgen = childPositions(5, offset=firstgen[i], length=0.5, separation=pi/15);
			    for (int j=0; j<secondgen.length; ++j)
			      {
			        draw((firstgen[i] -- secondgen[j]));
			        dot(secondgen[j]);
			      }
			    label("$\ensuremath{\underbrace{}_{}}$", firstgen[i]+(0,-2/3));
			    label("\scalebox{0.8}{$x-2$}", firstgen[i]+(0,-3/4));
			  }
		\end{asy}
		\caption{}
		\label{fig:trees}
	\end{figure}
\end{minipage}

\begin{lem}\label{lem:treecov}
For any graph $G=(V,E)$ with $\abs{V}=n,$ there exists a covering of $V$ with $\le k$ connected subgraphs of size $\le 2\floor{\frac n{k+1}}+1.$
\end{lem}

\begin{proof}
We use strong induction on $(k,n)$ (with the lexicographic order).
%(any linear ordering of $\mathbb N^2$ under the componentwise order will do). 
For the base case, when $k=1,$ $n \le 2 \floor{\frac{n}{2}} + 1 = 2\floor{\frac n{k+1}}+1$ so the graph itself is a cover, and when $n \le k,$ $2\floor{\frac n{k+1}}+1 = 2 \cdot 0 + 1 = 1$ so the set of vertices is a cover.

For the inductive step, assume the statement is true for all $k'<k, n'<n,$ with $1 < k < n.$ We want to show that the statement is true for $(k,n).$ 

As $n>k,$ $\floor{\frac n{k+1}}+1 \ge 1+1 = 2.$ Let $T$ be a rooted spanning tree of $G.$ Apply Lemma~\ref{lem:sector} with $x=\floor{\frac n{k+1}}+1<n$ on $T$ to get $S$ and $v.$ We have
\begin{align*}
\abs{S} \le 2x-1 &= 2\left( \floor{\frac n{k+1}}+1 \right) - 1
= 2\floor{\frac n{k+1}}+1.
\intertext{$G|_{S}$ will be one of the subgraphs in our covering. We now apply the inductive hypothesis on $G' = (V', E') = G|_{(V-S) \cup \left\{v\right\}}$ and $k' = k-1$ to get the rest of the subgraphs. We have}
n' = \abs{V'} &= \abs{(V-S) \cup \left\{v\right\}} = \abs{V} - \abs{S} + 1 \\
&< n - x + 1 = n - \left( \floor{\frac n{k+1}}+1\right) + 1 \\
&= n - \floor{\frac n{k+1}} = \ceil{\frac {kn}{k+1}},
\intertext{so $n' < n.$ Finally,}
2\floor{\frac {n'}{k'+1}}+1 &\le 2\floor{\frac {\ceil{\frac {kn}{k+1}} - 1}{(k-1)+1}}+1 \le 2\floor{\frac {\frac {kn}{k+1}}k}+1\\
&= 2\floor{\frac n{k+1}}+1,
\end{align*}
as desired.

\end{proof}

The next two bounds also apply to the known gambler, since the cops can simulate the observation for the $1$-observed gambler if they are facing the known gambler.

\begin{thm}
A distributed version of $\wmw{1}$ can be used by $k$ cops against the $1$-observed gambler on any connected $n$-vertex graph $G$ to achieve expected capture time at most $\frac{3n}{k+1} + 1$.
\end{thm}

\begin{proof}
Assign every subgraph (sector) given by Lemma~\ref{lem:treecov} to the first few cops, and ignore all unassigned cops. Each cop should start the center of their subgraph, so they are distance at most $\floor{\frac{n}{k+1}}$ from any vertex in their sector.

Label the vertices in each sector arbitrarily with colors $c_1, \dots, c_{2\floor{\frac{n}{k+1}}+1}$, where some colors are not necessarily used in all sectors, and all vertices in a single sector have different colors. (Shared vertices between sectors might be given multiple colors.) Let the gambler appear at a vertex with the color $c_i$ on their observed move. Every cop will move to the vertex with color $c_i$ corresponding to their sectors (or they will stay put if there is no such vertex).

Let $q_{i}$ be the probability that the gambler is observed at a vertex labeled $c_i$, and define $S = \left\{i: 1\leq i \leq 2 \floor{\frac{n}{k+1}} + 1 \wedge q_i > 0 \right\}$. Then the expected capture time is at most 
$$\sum_{i \in S} q_{i} \left(\floor{\frac{n}{k+1}} + \frac{1}{q_i}\right) = \floor{\frac{n}{k+1}} + \abs{S} \leq 3 \floor{\frac{n}{k+1}} + 1.$$
\end{proof}

\begin{thm}\label{diststar}
A distributed version of $\wmw{1}$ can be used by $k$ cops against the $1$-observed gambler on the $n$-vertex star to achieve expected capture time at most $\frac{n-1}{k}+3$.
\end{thm}

\begin{proof}
All cops start at the center of the star. Partition the non-center vertices of the star into $k$ sectors of size at most $\ceil{\frac{n-1}{k}}+1$. As in the last proof, label the vertices in each sector with colors $c_1, \dots, c_{\ceil{\frac{n-1}{k}}+1}$, where all vertices in a single sector have different colors. The center of every sector should be colored $c_{\ceil{\frac{n-1}{k}}+1}.$

If the gambler is observed at a vertex with the color $c_i$, every cop will move to the vertices with color $c_i$ in their sectors (or they will stay put if there is no such vertex). If $q_{i}$ is the probability that the gambler appears at a vertex labeled $c_i$ and $S = \left\{i: 1\leq i \leq \ceil{\frac{n-1}{k}}+1 \wedge q_i > 0 \right\}$, then the expected capture time is at most
$$\sum \limits_{i \in S} q_{i} \left(1+\frac{1}{q_i}\right) = \abs{S}+1 \leq \ceil{\frac{n-1}{k}} + 2.$$
\end{proof}

\section{Throttling results} \label{s:throttling}
We define the following three versions of throttling number for the cop versus gambler game:

\begin{definition}
    The \emph{known gambler throttling number} of a graph $G,$ hereafter denoted $\kt G,$ is the minimum possible value of $k + T_{k}$ where $T_k$ is the minimum expected capture time of a known gambler that $k$ cops can guarantee.
\end{definition}

\begin{definition}
    The \emph{unknown gambler throttling number} of a graph $G,$ hereafter denoted $\ut G,$ is the minimum possible value of $k + U_k,$ where $U_k$ is the minimum expected capture time of an unknown gambler that $k$ cops can guarantee.
\end{definition}

\begin{definition}
    The $t$-\emph{observed gambler throttling number} of a graph $G,$ hereafter denoted $\ot[t]{G},$ is the minimum possible value of $k + O_{t,k}$ where $O_{t,k}$ is the minimum expected capture time of a $t$-observed gambler that $k$ cops can guarantee.
\end{definition}

We make two remarks that follow from the definitions before proving the main results in this section.

\begin{remark}
    Given a graph $G,$ $\ut G \ge \ot[s]{G} \ge \ot[t]{G} \ge \kt G$ for $s < t.$
    
    %The unknown gambler throttling number of a graph $G$ is greater than or equal to the known gambler throttling number of $G.$
\end{remark}

\begin{remark}
    For graphs $H \subseteq G$ on the same set of vertices, $\ut H \ge \ut G,~ \kt H \ge \kt G,$ and $\ot[t]{H} \ge \ot[t]{G}.$ 
    %The (un)known gambler throttling number of a graph $H \subseteq G$ is greater than or equal to that of $G.$
\end{remark}

\begin{lem} \label{thlower}
    For any graph $G$ on $n$ vertices, $\ut G \ge \ot[t]{G} \ge \kt G \ge 2 \sqrt n.$
    %The (un)known gambler throttling number of a graph $G$ on $n$ vertices is at least $2 \sqrt n.$
\end{lem}

\begin{proof}
    As mentioned by Geneson in Lemma 3 of \cite{jg2}, if the gambler uses the uniform distribution on a connected $n$-vertex graph, then the expected capture time is at least $n/k.$ Therefore, the throttling number is at least the minimum value of $\dfrac n k + k$ over all positive integer values for $k,$ which in turn is at least $2 \sqrt n,$ as desired.
\end{proof}

\begin{lem}
   On a graph $G$ with $n$ vertices, $\kt G \le \ot[t]G \le \ut G < 3.96944\sqrt n$ for sufficiently large $n.$
   %The unknown gambler throttling number of a graph $G$ on $n$ vertices for $n \ge 10^{11}$ is less than $3.96944\sqrt n.$
\end{lem}

\begin{proof}
    As shown by Geneson in Theorem 5 of \cite{jg2}, the expected capture time for the unknown gambler with $k$ cops on a connected $n$-vertex graph, where $n>k,$ is at most 
    $$\left(\dfrac 1{1-e^{-2}}-1\right) \left(\dfrac{6n}k + 1\right) + \dfrac{3n}k + 1.$$
    
    For $n$ sufficiently large, one can choose a $k$ around ${\sqrt{\left(\dfrac6{1-e^{-2}}-3\right)n}}$ so that 
    $$\left(\dfrac 1{1-e^{-2}}-1\right) \left(\dfrac{6n}k + 1\right) + \dfrac{3n}k + 1 + k < 3.96944\sqrt n,$$
    giving our bound.
\end{proof}

\begin{lem}
   On a graph $G$ with $n$ vertices, $\kt G \le \ot[1]{G} < \ceil{2\sqrt{3n}\,}$. 
   %The unknown gambler throttling number of a graph $G$ on $n$ vertices for $n \ge 10^{11}$ is less than $3.96944\sqrt n.$
\end{lem}

\begin{proof}
    As shown in Section \ref{s:dist}, the expected capture time for the $1$-observed gambler with $k$ cops on a connected $n$-vertex graph, where $n>k,$ is at most $\frac{3n}{k+1} + 1$.
    
    We can choose $k = \floor{\sqrt{3n}-\frac{1}{2}}$ so that 
    $$\frac{3n}{1+k} + 1 + k \le  \ceil{2\sqrt{3n}},$$
    giving our bound.
\end{proof}

%\begin{lemma}
%    For respective cycles $C_m$ and $C_{m-1}$ on $m$ and $(m-1)$ vertices, $\kt {C_m} \ge \kt {C_{m-1}}.$
%\end{lemma}

The next upper bound is nearly tight, using the lower bound from Lemma \ref{thlower}.

\begin{lem}
    For the path $P_n$ on $n$ vertices, $\kt{P_n} \le \ceil{2 \sqrt n\,}.$
\end{lem}

\begin{proof}
    Label vertices along the path $v_1, \ldots, v_n.$  Let $p_i$ be the probability that the gambler assigns to $v_i.$
    
    Distribute $k = \ceil{\sqrt n - \frac 1 2}$ cops evenly along the path at every $m=\ceil{\frac nk}$ vertex, such that a cop initially occupies $v_i$ iff $i \equiv 1 \pmod m.$
    
    All cops move in unison; at each turn, they all decide to move to the next vertex, or stay put. (For the sake of simplicity, if the last cop has to step beyond $v_n,$ she can stay put and we can ignore her prospect of catching the gambler, thus if anything overestimating the expected capture time.)  Thus at any time the cops occupy vertices $v_j, v_{j+m}, v_{j+2m}, \ldots, v_{j+(k-1)m}$ for some $1 \le j < m,$ and their probability of catching the gambler at this time is $$q_j = \sum_{x \equiv j \pmod m} p_x.$$
    
    We have here effectively reduced our game to one supercop on a path $P_m$ on $m$ vertices, with a probability of $q_i$ assigned to vertex $1 \le i \le m.$
    
    By Lemma 2.1 of \cite{gambler0}, our supercop can guarantee an expected capture time of $m$ by stepping only in the direction of $v_m.$  Thus the known gambler throttling number is at most $$k + m \leq \ceil{2\sqrt n \,}.$$
\end{proof}

\begin{cor}
    All $n$-vertex graphs $G$ with a Hamiltonian path have $\kt{G} \leq \ceil{2 \sqrt n\,}$.
\end{cor}

Again, we note that the upper bound in the last corollary about Hamiltonian paths is tight in the coefficient of $\sqrt{n}$. The next lemma allows us to derive a similar result for the unknown gambler on graphs with Hamiltonian cycles. 

\begin{lem}
    For the cycle $C_n$ on $n$ vertices, $\kt{C_n} \le \ot[t]{C_n} \le \ut{C_n} < 2.08037 \sqrt n$ for sufficiently large $n.$
\end{lem}

\begin{proof}
    We use a strategy similar to that of Lemma 4.4 of \cite{gambler0} and Theorem 5 of \cite{jg2}.  Let $v_1, \ldots, v_n$ be the vertices of the cycle. For a fixed $k$, evenly distribute $k$ cops around the cycle at vertices $v_{a_1}, \ldots, v_{a_k}.$ Each cop puts down a flag at their vertex. The gambler chooses his probability distribution; let $p_i$ be the probability he assigns to vertex $v_i.$  Then, the direction that the cops will take around the cycle in their subsequent turns is determined by a single coin flip.
    
    The cops now move in unison in their chosen direction for the next $m = \ceil{\frac nk}$ turns, until they each reach the flag planted by their colleague. (Some may reach the flag one turn before the rest; they will then stay put for one turn.) We call this sequence of $m$ moves an \textit{inning}.  The cops' strategy is to repeat innings until they win.
    
    At any given turn $1 \le j \le m,$ the probability of capturing the gambler is 
    $$q_j = \sum\limits_{i=1}^{k} p_{a_i+j-1},$$
    and conversely the gambler's probability of evading capture is $1 - q_j.$ 
    Thus, the probability of the gambler surviving a full inning of $m$ turns is 
    $$\prod\limits_{j=1}^{m} \left(1 - q_j\right) = \prod\limits_{j=1}^{m} \left(1 - \sum\limits_{i=1}^{k} p_{a_i+j-1}\right).$$
    
    %The gambler's strategy is to maximize this probability, i.e. when $q_1 = q_2 = \ldots = q_m;$ for the simplicity we can assume this means $p_1 = p_2 = \ldots = p_n = \frac 1 n.$  Thus the gambler's maximal chance of surviving $m$ turns becomes
    %$$\prod\limits_{j=1}^{m} \left(1 - q_j\right) = \left(1 - k\frac 1n\right)^m = %\left(1 - \frac kn\right)^\ceil{\frac nk}.$$
    
    As every vertex is covered by a cop, we have
    $$\sum_{j=1}^m (1-q_j) = \sum_{j=1}^{m} 1 - \sum_{j=1}^{m} q_j \le m-1,$$
    thus 
    $$\prod\limits_{j=1}^{m} \left(1 - q_j\right) \le \left(\frac{m-1}m\right)^m \le \frac 1 e.$$
    Thus, the probability of capture per inning is at least $\left(1-\frac 1 e\right),$ and the expected number of unsuccessful innings for the cops is at most
    $$\frac 1 {1 - \frac 1e} - 1.$$
    
    Since the cops flipped a coin, the expected number of turns in a successful inning is at most $\frac m2,$ so 
    $$\ut{C_n} \le k+\left(\frac 1 {1 - \frac 1e} - 1 \right)m + \frac m2.$$
    
    For all $n$ sufficiently large, one can choose a $k$ around $\sqrt{\left(\frac 1 {1-\frac 1e} - \frac 12 \right) n}$ so that 
    $$\ut{C_n} \le k+\left(\frac 1 {1 - \frac 1e} - 1 \right)m + \frac m2 \le 2.08037\sqrt n,$$
    as desired.
\end{proof}

\begin{cor}
    All graphs $G$ on $n$ vertices with a Hamiltonian cycle have $\kt G \le \ot[t]G \le \ut{G} \le  2.08037 \sqrt n$ for sufficiently large $n.$
\end{cor}

Note that the coefficient of the upper bound in the last corollary is not tight with the lower bound in Lemma \ref{thlower}, but the coefficient of the upper bound in the theorem below is tight.

\begin{thm}
All graphs $G$ on $n$ vertices with a universal vertex have $\kt G \le \ot[1]{G} < \ceil{2 \sqrt n\,}+3$.
\end{thm}

\begin{proof}
We use $k = \ceil{\sqrt n - \frac 1 2}$ cops with the same strategy as in Theorem \ref{diststar}.
\end{proof}

We also find nearly tight bounds on the expected capture time and throttling number for the unknown gambler on complete graphs.

\begin{thm}
On the complete graph with $n$ vertices, the expected capture time for the unknown gambler versus $k$ cops is at most $1 + \frac{n}{k}$.
\end{thm}

\begin{proof}
The cops choose arbitrary initial positions. On each turn, the cops choose a subset of $k$ vertices uniformly at random and move to those vertices (or stay put if they are already there). The probability of capture on each turn after the first is 
$$\frac{\binom{n-1}{k-1}}{\binom{n}{k}}\sum\limits_{i = 1}^{n} p_{i} = \frac{k}{n},$$
implying the expected capture time is at most $1 + \frac{n}{k}$.
\end{proof}

Since the uniform gambler has expected capture time of at least $\frac{n}{k}$ versus $k$ cops, the last bound is off by at most $1$ from the actual value.

\begin{cor}
For the complete graph $K_{n}$ on $n$ vertices, $\ut{K_{n}} \le  \ceil{2 \sqrt n\,}+1$.
\end{cor}

\section{Conclusion}
In this paper, we defined the observed gambler, which is a more realistic version of the gambler for the application of anti-incursion algorithms. We also bounded the expected capture time and throttling numbers for the observed gambler, as well as the known and unknown gamblers that were defined in \cite{gambler0}. 

One of the strategies that we used against the $t$-observed gambler, $\kw{t}$, was based on Komarov and Winkler's strategy against the known gambler from \cite{gambler0}. We showed that $\kw{t}$ works almost as well against the $t$-observed gambler as it does against the known gambler for $t = \omega(n^2)$, but that $\kw{t}$ does not work substantially better than observing a single move and going to the observed vertex when $t = o(\sqrt{n})$. For $t$ between these ranges, it is an open problem to determine how $\kw{t}$ performs against the $t$-observed gambler.

Another natural variant of the gambler is an adversary that is visible for the first $t$ turns of the game and invisible after that point, where the cop knows no other information about the gambler's distribution. For the case $t = 1$, call this the \emph{once-visible} gambler. We show that the once-visible gambler has the same expected capture time on the star as the known gambler. Note that the next proof is almost the same as Theorem \ref{1-obs}, with the only difference being the expected capture time when the gambler is observed at the cop's initial vertex.

\begin{thm}
    Against the once-visible gambler on a connected $n$-vertex graph of radius $r$, the $\wmw{1}$ strategy has expected capture time at most $n+r-1$.
\end{thm}

\begin{proof}
Let $G$ be a connected $n$-vertex graph of radius $r$ and suppose that the cop has initial position at a center vertex $v_c$ of $G$, so every other vertex in $G$ has distance at most $r$ to $v_c$. If $p_i > 0$ and the cop chooses final vertex $v_i$, then the expected capture time is at most $r+\frac{1}{p_i}$. If the gambler appears at vertex $v_c$ on the first turn of the game, then capture occurs on the first turn. 

Let $S = \left\{i: 1\leq i \leq n \wedge p_i > 0 \wedge i \neq c  \right\}$, the set of indices of all vertices besides $v_c$ with positive probability. The expected capture time is at most 
$$p_c+\sum\limits_{i \in S} p_i \left(r+\frac{1}{p_i}\right) = r + \abs{S} \leq r+n-1.$$
\end{proof}

The last bound gives us the next two corollaries:

\begin{cor}\label{upgen}
$\wmw{1}$ has expected capture time at most $3n/2-1$ against the once-visible gambler on any connected $n$-vertex graph. 
\end{cor}

\begin{cor}
If both cop and gambler use best play, the expected capture time for the once-visible gambler on the $n$-vertex star is exactly $n$. 
\end{cor}

We finish with two questions that are not covered by our results.

\begin{enumerate}
\item For the $t$-observed gambler, what expected capture time can be achieved if up to $c$ of the observations can be wrong?
\item Suppose that we only have $m$ turns to catch the gambler on some connected $n$-vertex graph. What is the maximum probability of capture (assuming best play)? 
\end{enumerate}

The second question was also asked in CrowdMath 2017, where a Dijkstra-like algorithm was found for the known gambler \cite{crowdmath}, but no algorithm was found for the unknown gambler. Furthermore, no bounds for this problem were found for either type of gambler, even for specific families of graphs. Note that $m$ must be at least the radius or else the gambler sits out of range.

\section*{Acknowledgments}
Thanks to Tanya Khovanova, Slava Gerovitch, Pavel Etingof, Claude Eicher, the MIT Math Department, and the MIT PRIMES program for providing us with the opportunity to work together on this project.

\printbibliography
\end{document}